%% file: multistream.tex
\documentclass{sig-alternate}
\usepackage{amsmath}
\usepackage[ruled,vlined,boxed]{algorithm2e}
\usepackage{graphicx}
\usepackage{xspace}
\usepackage[center,tight]{subfigure}
\def\s{{S}\xspace}
\def\v{{\rm v}\xspace}
\def\sets{{\cal B}\xspace}
\def\rank{{\rm rank}\xspace}
\def\l{{\lambda}\xspace}
\def\avg{{\rm avg}\xspace}
\def\med{{\rm med}\xspace}
\def\2max{{\rm 2ndMax}\xspace}

\def\eps{\varepsilon}\xspace
\def\expb{\textsc{ExponentialBucket}\xspace}
\def\varb{\textsc{VariableBucket}\xspace}

\def\insert{\textsc{Insert}\xspace}
\def\compress{\textsc{Compress}\xspace}

\long\def\cut#1{{}}

\newcommand\comm[1]{\typeout{Used \string\comm...}\quad{\bf [[#1]]}\quad}

\newtheorem{theorem}{Theorem}
\begin{document}

\title{Untangling the Braid: \\ Finding Outliers in a Set of Streams}

\author{
\begin{tabular}{ccc}
Chiranjeeb Buragohain & Luca Foschini & Subhash Suri\\
Amazon.com & Dept. of Computer Science & Dept. of Computer Science \\
Seattle, USA & University of California & University of California\\
\texttt{chiran@amazon.com} & Santa Barbara, USA & Santa Barbara, USA\\
 & \texttt{foschini@cs.ucsb.edu} & \texttt{suri@cs.ucsb.edu}
\end{tabular}
}

\maketitle

\begin{abstract}

Monitoring the performance of large shared computing systems such as
the cloud computing infrastructure raises many challenging algorithmic
problems. One common problem is to track users with the largest
deviation from the norm (outliers), for some measure of
performance. Taking a stream-computing perspective, we can think of
each user's performance profile as a stream of numbers (such as
response times), and the aggregate performance profile of the shared
infrastructure as a ``braid'' of these intermixed streams. The
monitoring system's goal then is to untangle this braid sufficiently
to track the top $k$ outliers. This paper investigates the space
complexity of one-pass algorithms for approximating outliers of this
kind, proves lower bounds using multi-party communication complexity,
and proposes small-memory heuristic algorithms. On one hand, stream
outliers are easily tracked for simple measures, such as max or min,
but our theoretical results rule out even good approximations for most
of the natural measures such as average, median, or the quantiles. On
the other hand, we show through simulation that our proposed
heuristics perform quite well for a variety of synthetic data.

\end{abstract}


\section{Introduction}

Imagine a general purpose stream monitoring system faced with the task
of detecting misbehaving \emph{streams} among a large number of
distinct data streams.  For instance, a network diagnostic program at
an IP router may wish to highlight \emph{flows} whose packets
experience unusually large average network latency.  Or, a cloud
computing service such as Yahoo Mail or Amazon's Simple Storage
Service (S3), catering to a large number of distinct users, may wish
to track the quality of service experienced by its users.  The
performance monitoring of large, shared infrastructures, such as cloud
computing, provides a compelling backdrop for our research, so let us
dwell on it briefly. An important characteristics of cloud computing
applications is the sheer scale and large number of users: Yahoo Mail
and Hotmail support more than \emph{250 million users}, with each user
having several GBs of storage. With this scale, any downtime or
performance degradation affects many users: even a guarantee of
$99.9\%$ availability (the published numbers for Google Apps,
including Gmail) leaves open the possibility of a large number of
users suffering downtime or performance degradation.  In other words,
even a $0.1\%$ user downtime affects 250,000 users, and translates to
significant loss of productivity among users.  Managing and monitoring
systems of this scale presents many algorithmic challenges, including
the one we focus on: \emph{in the multitude of users, track those
  receiving the worst service}.

Taking a stream-computing perspective, we can think of each user's
performance profile as a stream of numbers (such as response times),
and the aggregate performance profile of the whole infrastructure as a
\emph{braid} of these intermixed streams. The monitoring system's goal
then is to untangle this braid sufficiently to track the top $k$
outliers. In this paper, we study questions motivated by this general
setting, such as ``which stream has the highest average latency?'', or
``what is the median latency of the $k$ worst streams?,'' ``how many
streams have their $95$th percentile latency less than a given
value?''  and so on.

These problems seem to require peering into individual streams more
deeply than typically studied in most of the extant literature. In
particular, while problems such as \emph{heavy hitters} and
\emph{quantiles} also aim to understand the statistical properties of
IP traffic or latency distributions of webservers, they do so at an
\emph{aggregate} level: heavy hitters attempt to isolate flows that
have large total mass, or users whose total response time is
cumulatively large.  In our context, this may be uninteresting because
a user can accumulate large total response time because he sends a lot
of requests, even though each request is satisfied quickly. On the
other hand, streams that consistently show high latency are a
cause for alarm. More generally, we wish to isolate flows or users
whose service response is bad at a finer level, perhaps taking into
account the entire distribution.

\subsection{Problem Formulation}

We have a set $\sets$, which we call a \emph{braid}, of $m$ streams 
$\{ \s_1, \s_2, \ldots, \s_m \}$, where the $i$th stream has 
size $n_i$, namely, $n_i = |\s_i|$.  We assume that the number 
of streams is large and each stream contains potentially an 
unbounded number of items; that is, $m \gg 1$ and $n_i \gg 1$, 
for all $i$. By $v_{ij}$, we will mean the value of the $j$th
item in the stream $\s_i$; we make no assumptions about $v_{ij}$
beyond that they are real-valued. In the examples mentioned above,
$v_{ij}$ represents the \emph{latency} of the $j$th request by the
$i$th user.  We formalize the misbehavior quality of a stream by an
abstract \emph{weight} function $\l$, which is function of the set of
values in the stream. For instance, $\l (\s)$ may denote the average
or a particular quantile of the stream $\s$. Our goal is to design
streaming algorithms that can estimate certain fundamental statistics
of the set $\{ \l (\s_1), \l (\s_2), \ldots, \l (\s_m ) \}$.

When needed, we use a self-descriptive superscript to discuss specific
weight functions, such as $\l^\avg$ for average, $\l^\med$ for median,
$\l^{\max}$ for maximum, $\l^{\min}$ for minimum etc. For instance, if
we choose the weight to be the \emph{average}, then $\l^\avg (\s)$
denotes the average value in stream $\s$, and 
\begin{displaymath}
\max \{ \l^\avg (\s_1), \l^\avg (\s_2), \ldots, \l^\avg (\s_m ) \}
\end{displaymath}
computes the \emph{worst-stream by average latency}. Throughout we
will focus on the one-pass model of data streams.

As is commonly the case with data stream algorithms, we must content
ourselves with \emph{approximate} weight statistics because even in
the single stream setting neither quantiles nor frequent items can be
computed exactly. With this in mind, let us now precisely define what
we mean by a guaranteed-quality approximation of high weight
streams. There are two natural and commonly used ways to quantify an
approximation: by \emph{rank} or by \emph{value}. (Recall that the
\emph{rank} of an element $x$ in a set is the \emph{number} of items
with value equal to or less than $x$.)

\begin{itemize}
\item {\bf Rank Approximation:} Let $\l$ be an arbitrary weight
  function (e.g. median), and let $\l_i = \l (\s_i)$ be the value of
  this function for stream $\s_i$.  We say that a value $\l'_i$ is a
  \emph{rank approximation} of $\l_i$ with \emph{error $E$} if the
  rank of $\l'_i$ in the stream $\s_i$ is within $E$ of the rank of
  $\l_i$. Namely,
	\[ | \rank(\l'_i , \s_i ) \;-\; \rank( \l_i , \s_i ) | \:\leq\: E ,\]
	where $\rank(x,\s)$ denotes the rank of element $x$ in stream
        $\s$ and $E$ is a non-negative integer.  Thus, if we are
        estimating the median latency of a stream, then $\l'_i$ is its
        rank approximation with error $E$ if 
	$$|\rank(\l'_i,\s_i ) \;-\; \lfloor |S_i|/2 \rfloor | \:\leq\: E .$$

\item {\bf Value Approximation:} Let $\l$ be an arbitrary weight
  function (e.g. median), and let $\l_i = \l (\s_i)$ be the value of
  this function for stream $\s_i$.  We say that $\l'_i$ is a
  \emph{value approximation} of $\l_i$ with \emph{relative error $c$}
  if
        \[ | \l'_i  \;-\; \l_i | \:\leq\: c \l_i .\]
\end{itemize}

While rank approximation often seems more appropriate for
quantile-based weights, and value approximation for average, they both
yield useful insights into the underlying distribution. For instance,
given any positive $\alpha$, at most $|\s| / \alpha $ items in $\s$
can have value more than $\alpha \l^\avg (\s)$.  Thus, a rank
approximation of $\l^\avg (\s)$ also localizes the relative position
of the approximation.  Conversely, a value approximation of the median
or quantile can be especially useful when the distribution is highly
clustered, making the rank approximation rather volatile---two items
may differ greatly in rank, but still have values very close to each
other.
Our overall goal is to estimate streams with large weights with
guaranteed quality of approximation: in other words, if we assert that
the worst stream in the set $\sets$ has median weight $\l^*$ then we
wish to guarantee that $\l^*$ is an approximation of $\max \{ \l^\med
(\s_1), \l^\med (\s_2), \ldots, \l^\med (\s_m ) \}$, either by rank or
by value. We prove possibility and impossibility results on what error
bounds are achievable with small memory.

\subsection{Our Contributions}

We begin with a simple observation that finding the \emph{top k}
streams under the $\l^{max}$ or $\l^{min}$ weight functions is easy:
this can be done using $O(k)$ space and $O(\log k)$ per-item
processing time. In the context of webservices monitoring
applications, this allows us to track the $k$ streams with worst
latency values. As is well known, however, statistics based on max or
min values are highly volatile due to outlier effects, and filtering
based on more robust weight functions such as quantiles or even
average is preferred.

We propose a generic scheme that can estimate the weight of any stream
using $O(\eps^{-2} \log U \log \delta^{-1})$ space (being $U$ the size of the range of the values in the streams), 
with rank error $\eps \sum_{i=1}^m n_i$ with
probability at least $1- \delta$. With this, we can report the weights 
of the top $k$ streams for any of the natural functions such as 
average, median, or other quantiles so that the rank error in the 
reported values is at most $\eps \sum_{i=1}^m n_i$.

One may object to the $\sum_{i=1}^m n_i$ term in the rank
approximation, and ask for a more desirable $\eps n_i$ error term so
that the error depends only on the size of an individual stream,
rather than the whole set of streams.  On pragmatic terms also, this
is justified because even for modest values (a few thousand streams,
each with a million or so items), the $\sum_{i=1}^m n_i$ error can
make the approximation guarantee worthless. In essense, our error
approximation is \emph{linearly worsening} with the number of streams,
which is not a very scalable use of space.

Unfortunately, we prove an impossibility result showing that achieving
$\eps n_i$ error in rank approximation requires \emph{space at least linear
in the number of streams (the braid size)}. Worse yet, our lower bound 
even rules out $\eps \tilde{n}$ rank approximation, where $\tilde{n}$ is the
\emph{average} stream size in the set. Thus, the space complexity is not an
artifact of \emph{rarity}, as is the case with the \emph{frequent item}
problem. In particular, we show that even if all streams in the braid
have size $\Theta( \tilde{n})$, achieving rank approximation $\eps \tilde{n}$
requires space $\Omega (m (\frac{1-2\eps}{1+2\eps})^{2 + \gamma})$, 
for any $\gamma > 0$, where $m$ is the number of streams in the braid.

Similarly, for the value approximation we show that estimating the
average latency of the worst stream in $\sets$ within a factor $t$
requires $\Omega (m/t^{2 + \gamma})$ space. Our lower bounds also rule out 
optimistic bounds even for highly structured and special-case streams. 
For instance, consider a \emph{round robin} setting where values arrive 
in a round-robin order over the streams, so at any instant the size 
of any two streams differs by at most one. One may have hoped that 
for such highly structured streams, improved error estimates should be 
possible. Unfortunately, a variant of our main construction rules out that
possibility as well.

In the face of these lowers bounds, we designed and implemented two
algorithms, \expb and \varb, and evaluated them for
a variety of synthetic data distributions. We use three quality
metrics to evaluate the effectiveness of our schemes:
\emph{precision} and \emph{recall}, which measures how many of the top $k$
captured by our scheme are true top $k$, \emph{distortion}, which
measures the average rank error of the captured streams relative to
the true top $k$, and \emph{average value error}, which measures absolute
value differences.
We tested our scheme on a variety of synthetic data distributions.
These data use a normal distribution of values
within the streams, and either a uniform or a normal distribution across
streams. In all these cases, our precision approaches 100\% for all
three metrics (average, median, 95th percentile), the distortion is between
1 and 2, and the average error is less than $0.02$.
The memory usage plot also confirms the theory that the size of the
data structure remains unaffected by the number and the sizes of the streams.

\subsection{Related Work}

Estimation of stream statistics in the one-pass model has received a
great deal of interest within the database, networking, theory, and
algorithms communities. While the one-pass majority finding algorithm
dates back to Misra and Gries~\cite{misra82finding}, and tradeoffs
between memory and the number of passes required goes back to the work
of Munro and Paterson~\cite{munro80selection}, a systematic study of
the stream model seems to have begun with the influential paper of
Alon, Matias and Szegedy~\cite{alon96space}, who showed several
striking results, including space lower bounds for estimating
frequency moments, as well as for determining the frequency of the
most frequent item in a single stream.  While some statistics such as
the average, min, and the max can be computed exactly and
space-efficiently, other more holistic statistics such as quantiles
cannot. Fortunately, however, several methods have been proposed over
the last decade to \emph{approximate} these values with bounded error
guarantees.  For instance, quantiles can be estimated with additive
error $\eps n$ using space $O( \eps^{-1} \log
n)$~\cite{greenwald01space} or space $O(\eps^{-1} \log
U)$~\cite{shrivastava04medians} where $n$ is the stream size and $U$
is the largest integer value for the stream items.
There is also a rich body of literature on finding frequent items,
top $k$ items, and heavy hitters~\cite{cormode2008ffi,charikar2004ffi,karp2003saf,manku2002afc,
metwally2005ecf,schweller2007rse}.

Schemes such as Counting Bloom filters~\cite{fan98summary} or
Count-Min sketches~\cite{cormode05improved} can be viewed as methods
for estimating statistics over multiple streams. In particular, these
methods are motivated by the need to estimate the sizes of \emph{large
  flows} at a router: in our terminology, these methods estimate the
\emph{aggregate sizes} of the top $k$ streams in the braid. By
contrast, we are interested in more refined statistics (e.g. top $k$
by the average value) that require peering into the streams, rather
than simply aggregating them.  Bonomi et al.~\cite{bonomi06beyond}
have extended Bloom filters to maintain not just the presence or
absence of a stream, but also some state information about the stream.
But this state information does not reflect any aggregate statistical
properties of the stream itself.   

The time-series data mining community has focused on finding
similar~\cite{faloutsos94fast} or dissimilar
sequences~\cite{keogh05hot} in a set of large time-series sequences.
But this work is not geared towards a one-pass stream setting and
hence assumes that we have $O(m)$ memory available where $m$ is the
number of streams.

\subsection{Organization}

Our paper is organized in five sections.
In Section~\ref{sec:lower}, we present our main theoretical results,
namely, the lower bounds on the space complexity of single-pass
algorithms for detecting outlier streams in a braid.
In Section~\ref{sec:algos}, we propose two generic space-efficient schemes
for estimating the top $k$ streams in a braid, and analyze
their error guarantees. In Section~\ref{sec:exp}, we discuss our
experimental results. Finally, we conclude with a discussion
in Section~\ref{sec:conc}.

\section{Space Complexity Lower Bounds}
\label{sec:lower}

In this section, we present our main theoretical results, namely,
space complexity lower bounds that rule out space-efficient approximation
of outlier streams in a fairly broad setting. We mentioned earlier
that for simple weight functions such as the max or min, one can
easily track the top $k$ streams, using just $O(k)$ space and
$O(\log k)$ per-item processing. (This is easily done by maintaining
a heap of $k$ distinct streams with the largest item values.)
Surprisingly, this good news ends rather abruptly: we show that even 
tracking top $k$ streams using the \emph{second largest} item is 
already hard, and requires memory proportional to the size of the
braid, $|\sets|$.
Similarly, we argue that while tracking streams with the maximum or
the minimum items is easy, tracking streams with the largest
\emph{spread}, namely, difference of the maximum and the minimum
items, requires linear space.
Our main result rules out even good approximation of most of the
major statistical measures, such as average, median, quantiles, etc.

We begin by recalling our formal definition of approximating the
outlier streams.  Suppose we wish to rank the streams in the braid
$\sets$ using a weight function $\l$. Without loss of generality,
assume that the top $k$ streams are indexed $1, 2, \ldots, k$; that
is, $\l_1 \geq \l_2 \geq \cdots \geq \l_k$.  We say that a stream
$S_i$ is \emph{approximately a top $k$ stream} if its $\l$-value is at
least as large as $\l_k$ within the approximation error range. For
instance, suppose we are using the median latency $\l^\med$, then
stream $S_i$ is a top $k$ stream with rank approximation $E$ if the
value of item with rank $|S_i|/2 + E$ (true median plus the rank
error) in $S_i$ is at least as large as $\l_k$. Similarly, one can
define the approximation for value approximation.  In the following,
we discuss our lower bounds, which are all based on the multi-party
communication
complexity~\cite{bar-yossef04information,yao79some,nisan97communication}.
All our lower bounds employ variations on a single construction, so we
begin by describing this general argument below.

\subsection{The Lower Bound Framework}
Our lower bounds are based on reductions from the \emph{multi-party
  set-disjointness} problem, which is a well-known problem in
communication complexity~\cite{nisan97communication}.  An instance
DISJ$_{m,t}$ of the multi-party set disjointness problem consists of
$t$ players and a set of items $A = \{1,2,\ldots, m\}$. The player
$i$, for $i=1,2,\ldots, t$, holds a subset $A_i \subseteq A$. Each
instance comes with a promise: either all the subsets $A_1, A_2,
\ldots A_t$ are pairwise disjoint, or they \emph{all} share a single
common element but are otherwise disjoint.  The former is called the
YES instance (disjoint sets), and the latter is called the NO instance
(non-disjoint sets).  The goal of a communication protocol is to
decide whether a given instance is a YES instance or a NO
instance. The protocol only counts the total number of bits that are
exchanged among the players in order to decide this; the computation
is free.  We will use the following result from communication
complexity~\cite{bar-yossef04information}: any one-way protocol (where
player $i$ sends a message to player $i+1$, for $i=1,2, \ldots,t$,
that decides between all YES instances and NO instance with success
probability greater than $1-\delta$, for any $0< \delta < 1$, requires
at least $\Omega(m/t^{1+\gamma})$ bits of communication, where recall
that $t$ is the number of players and $m$ is the size of the set and
$\gamma > 0$ is an arbitrarily small constant.

The idea behind our lower bound argument is to simulate a one-way
multi-party set-disjointness protocol using a streaming algorithm for
the top $k$ streams.  If the streaming algorithm uses a synopsis data
structure of size $M$, then we show that there is a one-way protocol
using $O(Mt)$ bits that can solve the $t$-party disjointness
problem. Because the latter is known to have a lower bound of $\Omega
(m /t^{1+\gamma})$, it implies that the memory footprint of the streaming
algorithm becomes $\Omega(m / t^{2+\gamma})$. The basic construction associates
a stream with each element of the set $A$; namely, the stream $S_i$ is
identified with the item $i \in A$.  See Fig.~\ref{fig:lower-bound}
for an example.  We initialize each stream with some values, and then
insert the remaining values based on the sets $A_i$ held by the $t$
players; these values depend on specific constructions.  The key idea
behind all our constructions is that the braid of streams is such that
approximating the top $k$ streams within the approximation range
requires distinguishing between the YES and NO instances of the
underlying set-disjoint problem. We present the details below as we
discuss specific constructions. We begin with our main result on the space
complexity of tracking the top $k$ streams under the most common
statistical measures, such as median, quantile, or average.

\subsection{Space Complexity of Ranking by Median or Average}

\begin{figure}
  \includegraphics[width=0.45\textwidth]{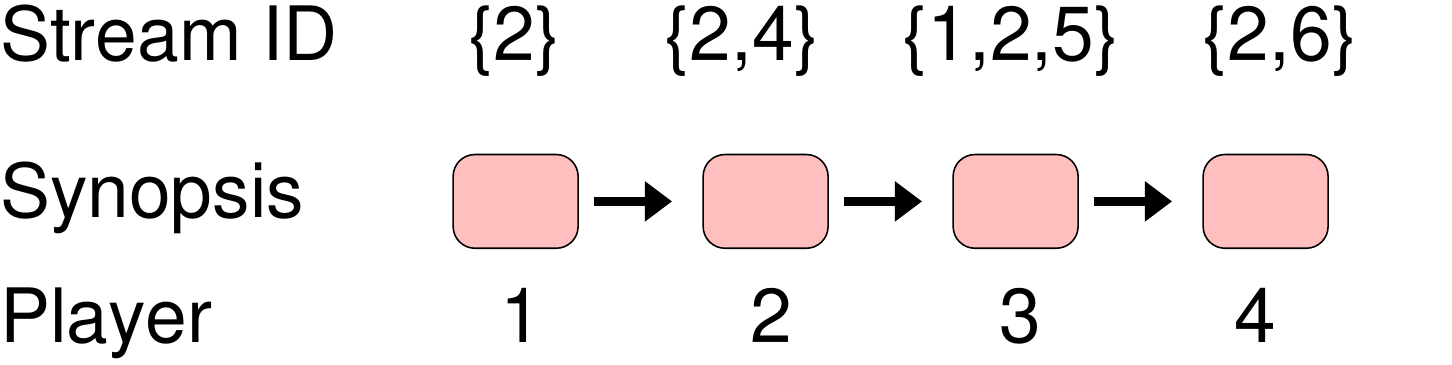}
  \caption{Illustration of a NO instance of DISJ$_{6,4}$, with four players.
	The subsets of the players are shown above their ids, and the values
    	inserted by them into the streams (following the protocol of
    Theorem~\ref{thm:lower-bound-median}) are shown in the Table~\ref{tbl:lower-bound}}.
  \label{fig:lower-bound}
\end{figure}

\begin{table}
  \begin{center}
  \begin{tabular}{|l|l|l|l|l|l|l|}
    \hline
    Stream ID & 1 &  2 &  3 &  4 & 5 & 6\\     \hline
    Player 1 & 0& 1&0&0&0&0\\     \hline
    Player 2 & 0& 1&0&1&0&0\\     \hline
    Player 3 & 1& 1&0&0&1&0\\     \hline
    Player 4 & 0& 1&0&0&0&1\\     \hline
  \end{tabular}
  \end{center}
  \caption{The values inserted by the players of Figure~\ref{fig:lower-bound}
	in the lower bound construction of Theorem~\ref{thm:lower-bound-median} are 
	shown in this table.}
  \label{tbl:lower-bound}
\end{table}

\begin{theorem}
Let $\sets$ be a braid of $m$ streams, where each stream has $\Theta (n)$
elements. Then, determining the top stream in $\sets$ by median value, 
within rank error $\eps n$ ($0 < \eps < 1/2$) requires space at least
$$\Omega\left(m\left(\frac{1-2\eps}{1+2\eps}\right)^{2+\gamma}\right).$$
for arbitrarily small $\gamma > 0$.  That is, finding the stream with
the maximum median latency, within additive rank approximation error
$\eps n$, requires essentially space linear in the number of streams.
\label{thm:lower-bound-median}
\end{theorem}
\begin{proof}
Suppose there exists a stream synopsis of size $M$ that can estimate
the latency of the maximum median latency stream within rank error
$\eps n$.  We now show a reduction that can use this synopsis to solve
the multi-party set disjointness problem using $O(M)$ bits of
communication.  Let $p$ be an integer, to be fixed later. We
initialize the synopsis by inserting $p$ items in each stream with
value 0.  The multi-party protocol then modifies the stream as
follows, one player at a time, from player 1 through $t$.  On his
turn, if player $j$ has item $i$ in its set, then it inserts $p$ items
of value 1 in stream $i$, for each $i \in A_j$.  If the player $j$
\emph{does not have} item $i$ in its set, then it inserts $p$ items of
value $0$ in the stream $i$.  (Recall that items of the ground set $A$
correspond to streams in our construction.)  Thus each player inserts
precisely $pm$ values to the streams, and in the end, each stream has
exactly $p+pt$ items in it.  An example of the running of this
protocol is shown in Fig.~\ref{fig:lower-bound} and the corresponding
values inserted by each player is tabulated in
Table~\ref{tbl:lower-bound}.

\begin{figure}
  \includegraphics[width=0.45\textwidth]{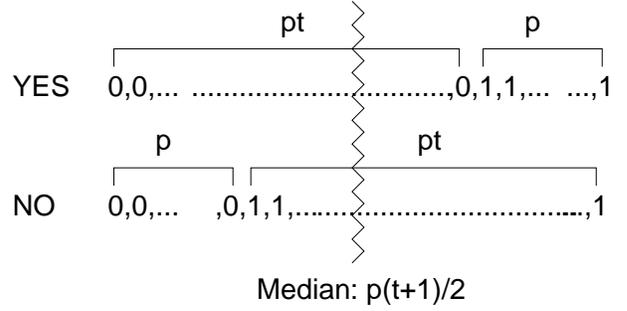}
  \caption{Distribution of values in streams for  YES and NO instances
    of  DISJ$_{m,t}$.  
For the YES instance, the median value is 0
    up to a rank error of $pt-p(t+1)/2$, while for the NO instance,
    the median value is 1 up to a rank error of $p(t+1)/2 -p$.}
  \label{fig:value-dist}
\end{figure}
After all the $t$ players are done, we output YES if the maximum
median latency among all streams is 0, and otherwise we return NO.  We
now reason why this helps decide the set-disjointness problem.
Suppose the instance on which we ran the protocol is a YES instance.
Then any stream has either all 0 values (this happens when index
corresponding to this stream is absent from all sets $A_j$), or it has
$p$ values equal to $1$ and $pt$ values equal to $0$ (because of the
disjointness promise, the index corresponding to this stream occurs in
precisely one set $A_j$ and that $j$ inserts $p$ copies of $1$ to this
stream, while others insert $0$s).  See Fig.~\ref{fig:value-dist}.
Therefore up to a rank error $p(t-1)/2$ the median latency of all the
streams is 0.  On the other hand, if this is a NO instance, then there
exists a stream that has $p$ 0 values and $pt$ values equal to $1$.
This stream, therefore, has a median latency of $1$ up to a rank error
$p(t-1)/2$.  Therefore our algorithm can distinguish between a YES and
a NO instance.

We may choose $p = n/(1+t)$, so that each stream has size $n$ and our
rank error is $n\frac{1}{2}\frac{t-1}{t+1}$.  Since the algorithm uses
$O(M)$ space and there are $t$ players, the total communication
complexity is $\Theta(Mt)$, which by the communication complexity
theorem is $\Omega(m/t^{1+\gamma})$.  Finally, solving the equality $\eps =
\frac{1}{2}\frac{t-1}{t+1}$ for $t$, we get the desired lower bound
that $M = \Omega\left(m\left(\frac{1-2\eps}{1+2\eps}\right)^{2+\gamma}\right)$.
This completes the proof.
\end{proof}

We point out that our stream construction is highly structured,
meaning that this lower bound rules out good approximation even for
very regular and balanced streams.  In particular, the difficulty of
estimating the maximum median latency is not a result of rarity of the
target stream: indeed, all streams have equal size.  Moreover, the
construction can be implemented in a way so that items are inserted
into the streams in a round-robin way (see
Table~\ref{tbl:lower-bound}).   Therefore, the construction
is also not dependent on a pathological spikes in stream
population. Thus, even under very strict ordering of values in the
streams, the problem of determining high weight streams remains hard.

It is easy to see that the construction is easily modified to 
prove similar lower bounds for other quantiles.
The same construction also shows a space lower bound for determining 
the maximum \emph{average latency} stream. Simply observe that the
average latency for the NO instance is $1/(1+t)$, while the average
latency for the YES instance is $t/(1+t)$.  Therefore any $t$-approximation 
algorithm for the average measure requires space at least $\Omega(m/t^{2 + \gamma})$,
for any $t \geq 2$ and $\gamma > 0$. 

\begin{theorem}
Determining the top stream by average value within relative error at most 
$t$ requires at least $\Omega(m/t^{2 + \gamma})$ space, where $m$ is 
the number of streams in the braid, $t \geq 2$ and $\gamma > 0$.
\label{thm:lower-bound-avg}
\end{theorem}

\subsection{Lower Bound for Second Largest}

Surprisingly, similar constructions also show that even minor
variations of the easy case (finding the stream with the largest or
the smallest extremal value) make the problem provably hard.  In
particular, suppose we want to track the stream with the maximum
\emph{second largest} value. Let us denote this weight function as
$\l^{2max}$. Our proof below proves a space lower bound for even
approximating this. In particular, we say that an streaming algorithm
finds the second largest-valued stream with \emph{approximation factor
$c$} if it returns a stream whose rank by the second largest value
is at most $2c$, for any integer $c \geq 1$.  Note that this
definition of approximation is one sided, because the approximate
value returned always has rank $2c \geq 2$. (Of course,
allowing $c < 1$ trivializes the problem because then we can
always use the max value instead of the second largest.)
Then, we have the following.

\begin{theorem}
Determining the top stream by second-largest value within approximation 
factor $t$ requires at least $\Omega(m/t^{2+\gamma})$ space, for any $\gamma > 0$,
where $m$ is the number of streams in the braid and $t \geq 2$ is an integer.
\end{theorem}
\cut{
\begin{proof}
In this case, suppose there exists a streaming algorithm for
the top stream by second largest value problem using space $M$, and 
consider the following reduction from an instance of the $t$-party 
set-disjointness problem. Each player $j$, for $j=1,2,\ldots, t$,
in turn adds items to stream using the following rule:
player $j$ inserts values $j+1$ and $1/(j+1)$ in the stream $S_i$
for each $i$ in its set $A_j$. It inserts no values in any other
stream.

At the end, we return YES to the set-disjointness problem if our
streaming algorithm determines that $\lambda^{\2max} < 1$,
and NO otherwise.  We now reason its correctness.
If the $t$-party instance is a YES instance, then any stream $i$
either contains the value 0, or it contains values $p+1$ and $1/(p+1)$, 
for some $1\leq p \leq t$.
Thus, the top stream by the second largest value has value less than 1. 
On the other hand, if this is a NO instance, then there exists a stream 
that includes all the values $1/(t+1), 1/t, \ldots, t, t+1$, whose
second largest value is $t$. Therefore if our algorithm has an approximation 
ratio better than $t/2$, it can distinguish between YES and NO instances. 
Because our protocol requires sending $M$-size synopsis to $t$ players,
the total communication complexity is $\Theta(Mt)$, which by the 
lower bound on set-disjointness is at least $\Omega(m/t)$.
Therefore, determining the top stream by second largest must
require at least $\Omega(m/t^2)$ space, and this completes the proof.
\end{proof}

\comm{This is an alternate proof which to mind is better since it makes sure
  that all streams have equal number of items and hence the
  approximation factor of $t$ makes sense for both YES and NO instances.}
}

\begin{proof}
In this case, suppose there exists a streaming algorithm for the top
stream by second largest value problem using space $M$, and consider
the following reduction from an instance of the $t$-party
set-disjointness problem. Each player $j$, for $j=1,2,\ldots, t$, in
turn adds items to stream using the following rule: player $j$ inserts
values $j+1$ and $1/(j+1)$ in the stream $S_i$ for each $i$ in its set
$A_j$.  For all other streams that it does not hold, it inserts two 0
values in that stream.  Therefore every player inserts $2m$ values
into the braid and at the end of the protocol, each stream has exactly
$2t$ values.

At the end, we return YES to the set-disjointness problem if our
streaming algorithm computes the value of the top stream as less than
1, and NO otherwise.  We now reason its correctness.  If the $t$-party
instance is a YES instance, then any stream $i$ either contains the
values 0, or it contains values $0, p+1$ and $1/(p+1)$, for some
$1\leq p \leq t$.  Thus, the top stream by the second largest value
has value less than 1 up to an approximation factor of $t/2$.  On the
other hand, if this is a NO instance, then there exists a stream that
includes all the values $1/(t+1), 1/t, \ldots, t, t+1$, whose second
largest value is $t$.  Therefore if our algorithm has an approximation
ratio better than $t/2$, it can distinguish between YES and NO
instances.  Because our protocol requires sending $M$-size synopsis to
$t$ players, the total communication complexity is $\Theta(Mt)$, which
by the lower bound on set-disjointness is at least $\Omega(m/t^{1+\gamma})$.
Therefore, determining the top stream by second largest must require
at least $\Omega(m/t^{2+\gamma})$ space, and this completes the proof.
\end{proof}

\subsection{Lower Bound for Spread}


We next argue that while tracking the top stream with the largest
or the smallest value is possible, tracking the top stream with
the largest \emph{spread}, namely, $\max(S) -\min(S)$ is not 
possible without linear space.

\begin{theorem}
Determining the top stream by the spread requires at least
$\Omega(m)$ space, where $m$ is the number of streams in the braid.
\end{theorem}
\begin{proof}
Let us consider an instance of \textsc{DISJ}$_{m,t}$ and a streaming
algorithm with a synopsis of size $M$ which can determine the stream
with maximum spread.  

In this case, we can use a 2-party set-disjointness lower bound.
Let us call the two players, ODD and EVEN.
We begin by inserting a single value 0 in each of the $m$ streams.
First, the ODD player inserts the value $-1$ into each stream $S_i$
for which $i$ is in its set $A_{\textrm{ODD}}$.
Next, the EVEN player inserts the value $+1$ into each stream $S_i$
for which $i$ is in its set $A_{\textrm{EVEN}}$.
Clearly, the top stream by the maximum spread has spread 1, then
the sets of ODD and EVEN are disjoint, and so this is a YES instance.
Otherwise, the top stream has spread 2, and this is a NO instance.
The synopsis size of the streaming algorithm, therefore, is at
least $\Omega (m)$. This completes the proof.
\end{proof}

This finishes the discussion of our lower bounds. The main conclusion
is that approximating the top $k$ streams either by average value,
within any fixed relative error, or by any quantile, within
a rank approximation error of $\eps n_i$, is not possible, where
$n_i$ is the size of the top stream. In fact, the lower bound
even rules out the rank approximation within error of $\eps \tilde{n}$,
where $\tilde{n}$ is the \emph{average} size of the streams in the braid.

In the following section, we complement these lower bounds by describing 
a scheme with a worst-case rank approximation error 
$\sum_{i=1}^m \eps n_i$, using roughly $O(\eps^{-2} \log U)$ space.

\section{Algorithms for Braid Outliers}
\label{sec:algos}

We begin with a generic scheme for estimating top $k$ streams, and
then refine it to get the desired error bounds.
The basic idea is simple. Without loss of generality, suppose
the items (values) in the streams come from a range $[1, U]$. We subdivide 
this range into subranges, called \emph{buckets}, that are pairwise 
disjoint and cover the entire range $[1,U]$. All stream entries
with a value $v$ are mapped to the bucket that contains $v$. Within
the bucket, we use a sketch, such as the Count-Min sketch, to
keep track of the number of items belonging to different streams.
With this data structure, given any value $v$ and a stream index $i$,
we can estimate how many items of stream $S_i$ have values
in the range $[1,v]$. This is sufficient to estimate various
streams statistics such as quantiles and the average.
We point that this estimation incurs two kinds of error:
one, a sketch has an inherent error in estimating how many
of the items in a bucket belong to a certain stream $S_i$, and
two, how many of those are less than a value $v$ when $v$ is some
arbitrary value in the range covered by the bucket.
For the former, we simply rely on a good sketch for frequency
estimation, such as the Count-Min, but for the latter, we
explore two options, which control how the bucket boundaries are
chosen.

The first algorithm, called the \textsc{ExponentialBucket} algorithm, 
splits the range into pre-determined buckets, with boundary
$[\ell_b, r_b]$ such that the ratio $\ell_b/r_b$ is constant.  
This ensures that the relative \emph{value} error of our approximation
is bounded by $(r_b-\ell_b)/r_b$.  However, the pre-determined
buckets is unable to provide a non-trivial rank approximation error bound. 
Our second scheme, therefore, takes a more sophisticated approach to
bucketing, and \emph{adapts} the bucket boundaries to data, so as
to ensure that roughly an equal number of items fall in each bucket.

Before describing these algorithms in detail, let us first quickly
review the key properties of our frequency-estimation data structure,
Count-Min sketch, because we rely on its error analysis.
The Count-Min (CM) sketch~\cite{cormode05improved} is a randomized
synopsis structure that supports approximate count queries over data streams.  
Given a stream of $n$ items, a CM sketch estimates the frequency of 
any item up to an \emph{additive} error of $\eps n$, with (confidence)
probability at least $1-\delta$. The synopsis requires space
$O(\frac{1}{\eps}\log\frac{1}{\delta})$.  The per-item processing time
in the stream is $O(\frac{1}{\delta})$.  We shall use the Count-Min data structure 
as a building block of our algorithms, but any similar sketch with
the frequency estimation bound will do for our purpose.

\subsection{The Exponential Bucket Algorithm}

The {\expb} scheme divides the value range $[1,U]$ into roughly
$\log_{1 + \rho} U$ buckets. The first bucket has the range
$[1, 1+\rho)$, the second one has range $[1+\rho, (1 + \rho)^2 )$,
and so on.  There are a total of $\log_{1+ \rho} U$ buckets, with the
last one being $[U/(1+\rho), U]$. Note that the ranges are semi-closed, 
including the left endpoint but not the right. Only the last bucket is 
an exception, and includes both the endpoints. We will say that the $i$th
bucket has range $[(1+\rho)^i, (1+\rho)^{i+1}]$, with the first
bucket being labeled the $0$th bucket.

\begin{algorithm}
\SetLine
\ForEach{item  $v_{ij}$ in the braid}{
  bucketId $\leftarrow \frac{\log v_{ij}}{\log(1+\rho)}$ \;
  insert($i$) into CMSketch(bucketId)\;
}
\caption{\textsc{ExponentialBucket} Algorithm}
\end{algorithm}

A stream entry $v$, is associated with the unique bucket containing
the value $v$.  For every bucket, we maintain a CM sketch to count
items belonging to a stream id $i$. In particular, given an item
$v_{ij}$ ($j$-the value in stream $i$) in the braid, we first
determine the bucket
$\lfloor \log_{1+\rho} v_{ij} \rfloor$ containing this
item, and then in the CM sketch for that bucket, we insert the stream
id $i$.  Because there are $\log_{1+\rho} U$ buckets, and each
bucket's Count-Min sketch requires $O(\eps^{-1} \log \delta^{-1} )$
memory, the total space needed is $O(\eps^{-1} \log_{1+\rho} U \log
\delta^{-1} )$.

Let us now consider how to estimate the number of values belonging to
stream $i$ in a particular bucket $b$. The Count-Min sketch can
estimate the occurrences of stream $i$ in this bucket with additive error
at most $\eps n(b)$, where $n(b)$ is the number of values from all 
streams that fall into bucket $b$. Now suppose we want to approximate
the median value for a stream $i$. We first estimate $n_i = \sum_bn_i(b)$,
the total number of items in stream $i$ over all the buckets. The
error in estimating $n_i$ is given by the sum of individual errors in
each bucket

\begin{equation}
  {\rm  error}(n_i) = \sum_b \eps n(b) = \eps n.
\end{equation}
Then we find the bucket $B$ such that
\begin{equation}
  \sum_b^{B-1}n_i(b) \:\:\leq\:\: n_i/2 \:\:\leq\:\:  \sum_b^Bn_i(b).
\end{equation}

Then we report the left boundary of the bucket $B$ as our estimate for the
median value for stream $i$.  
We have the following theorem.

\begin{theorem}
The {\expb} is a data structure of size
$O(\eps^{-1} \log_{1+\rho} U \log \delta^{-1} )$ that, with probability at 
least $1-\delta$, can find the top $k$ streams in a set of $m$ streams by 
average, median, or any quantile value.
\end{theorem}

The {\expb} scheme is simple, space-efficient, and easy to implement,
but unfortunately one cannot guarantee any significant rank or value approximation 
error with this scheme. For instance, in the worst-case, all items could fall
in a single bucket, giving us only the trivial rank error of $n$.
Similarly, it could also happen that all elements tend to fall into
the two extreme buckets, and the $\eps n$ error in size estimation may
cause us to be incorrect in our value approximation by $\Theta (|U|)$.
Thus, we will use \expb only as a heuristic whose main virtue is simplicity,
and whose practical performance may be much better than its worst-case.
In the following, we present a more sophisticated scheme that 
adapts its bucket boundaries in a data-dependent way to yield
a rank approximation error bound of $\eps n$.

\subsection{The Variable Bucket Algorithm}

The basic building block of {\varb} is the $q$-digest data structure~\cite{shrivastava04medians},
which is a deterministic synopsis for estimating the quantile of a data 
stream. At a high level, given a stream of $n$ values in the range $[1,U]$, 
the q-digest partitions this range into $O(\rho^{-1} \log U)$ buckets such 
that each bucket contains $O(\rho n)$ values.  This synopsis allows us to 
estimate the $\phi$-th quantile of the value distribution in the stream up 
to an additive error of $\rho n$ using space $O(\rho^{-1} \log U)$. 
We briefly describe the q-digest data structure below with its important 
properties, and then discuss how to construct our {\varb} structure
on top of it.  Throughout we shall assume that $U$ is a power of 2 
for simplicity.

\subsubsection{Approximate Quantiles Through  q-digest}

The q-digest divides the range $[1,U]$ into $2U-1$ tree-structured buckets.
Each of the lowest level (zeroth level) bucket spans just a single value,
namely, $[1,1],[2,2],\ldots,[U,U]$.  The next level bucket ranges are
$[1,2],[3,4],\ldots,[U-1,U]$, the one after that 
$[1,4],[5,8],\ldots,[U-3,U]$ and so on until the highest level bucket
span the entire range $[1,U]$.  In general the buckets at level $\ell$ are 
of the form 
$$[2^\ell(2^i-1)+1,2^\ell2^i)],$$ 
where $\ell=0,1,2,\ldots$ and $i=0,1,2\ldots$.  These buckets can be 
naturally organized in a binary tree of depth $\log U$ as shown in Fig.~\ref{fig:q-digest}.  For example, the 
bucket $[1,4]$ has two children: $[1,2]$ and $[3,4]$, while $[1,4]$ itself is 
the left child of $[1,8]$.  Every bucket contains in integer counter which 
counts the number of values counted within that bucket.  Note that the 
buckets in a q-digest are not disjoint: a single bucket overlaps in range with
all its children and descendants.  A q-digest with error parameter
$\rho$ consists of a small subset (size $O(\rho^{-1} \log U )$) of all
possible $2U-1$ buckets.

\begin{figure}
  \begin{center}
    \includegraphics[width=0.45\textwidth]{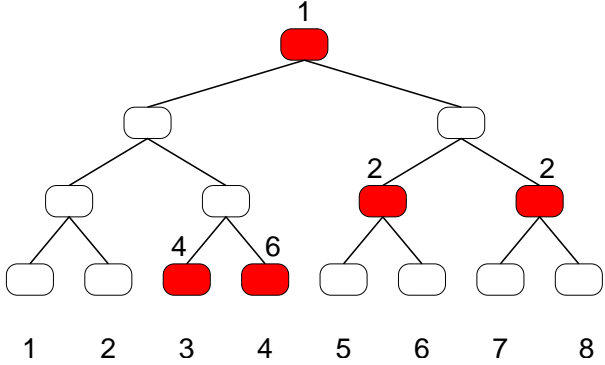}
  \end{center}
  \caption{A q-digest formed over the value range [1,8].  The complete
    tree of buckets is shown and the buckets that actually exist on
    the q-digest are highlighted in red.  The numbers next to the
    filled buckets show how many values were counted within that
    bucket.  For this q-digest, $n = 15, U = 8$ and the threshold $\lfloor
    n\rho/\log U\rfloor = 3$ (see
    eqns.~\ref{eqn:error},~\ref{eqn:memory}).}
\label{fig:q-digest}
\end{figure}

Intuitively a q-digest has many similarities to a equi-depth histogram: 
the buckets correspond to the histogram buckets and we
strive to maintain the q-digest such that all buckets have roughly
equal counts.  The memory footprint of the q-digest is proportional to
the number of buckets.  Therefore to reduce memory consumption, we can
take two sibling buckets and \emph{merge} them with the parent bucket.
The merge is done by deleting both the children and then adding their
counts to the parent bucket.  The merge step loses information, since
the counts of both the children are lost, but reduces memory
consumption.

Formally speaking, a q-digest with error parameter $\rho$, is a subset
of all possible buckets such that it satisfies the following q-digest
invariant.  Suppose that the total number of values counted within the
q-digest is $n$.  Then any bucket $b$ in the q-digest satisfies the
following two properties:
\begin{eqnarray}
  \rm{count}(b) & \leq & \left\lfloor \frac{n\rho}{\log U}
  \right\rfloor \label{eqn:error}\\ 
\rm{count}(b) + \rm{count}(b_p) + \rm{count}(b_s) & > & \left\lfloor
\frac{n\rho}{\log U} \right\rfloor 
\label{eqn:memory}
\end{eqnarray}
where $b_p$ is the parent bucket of $b$ and $b_s$ is the sibling
bucket of $b$.  The first property (\ref{eqn:error}) ensures that none
of the buckets are too heavy and hence attempts to preserve accuracy
being lost by merging too many buckets.  The only exception to this
property is the leaf bucket, which due to integer value assumption cannot
be divided any further.  The second property (\ref{eqn:memory})
ensures that the total values counted in a bucket, its sibling and
parent are not too few; therefore it encourages merging of buckets to
reduce memory.  The only exception to this property is the root bucket
since it does not have a parent.

The q-digest supports two basic operations: \insert and \compress.
Below, we show how we extend this q-digest structure to implement our
\varb algorithm and how the basic operations work.

\subsubsection{Variable Buckets Using q-digest}
\begin{algorithm}
\SetLine
\ForEach{item $ v_{ij}$ in the braid}{
  \If{bucket $b = [v_{ij},v_{ij}]$ does not exist} {
    create bucket $b$ \;
  }
  insert $i$ into CMSketch($b$)\;
  increment count($b$) \;
  \If{$b$ and its parent and sibling violates q-digest property}{
    \compress
  }
}
\caption{\varb \insert Algorithm}
\end{algorithm}

\begin{algorithm}
\SetLine
\For{$\ell=0$ {\KwTo} $(\log U-1)$} {
  \ForEach{ bucket $b$ in level $\ell$} {
    \If{$\textrm{count}(b) + \textrm{count}(b_p) + \textrm{count} (b_s) 
      < \lfloor n  \rho/\log U \rfloor$ } {
      count($b_p$) $\leftarrow$ count($b$) + count ($b_s$) \;
      CMSketch($b_p$) $\leftarrow$  
      CMSketch($b_p$) $\cup$  CMSketch($b$) $\cup$ CMSketch($b_s$) \;
      delete $b$ and $b_s$ \;
    } 
  } 
}
\caption{\varb \compress Algorithm}
\end{algorithm}

The \varb algorithm can be understood as a derivative of the q-digest
data structure.  In the basic q-digest, we divide the input values
into $\rho^{-1} \log U$ buckets and in each bucket we count the values in
that bucket using a simple counter.  In the \varb synopsis, we
augment this simple counter by an CM sketch of size
$O(\eps^{-1}\log \delta^{-1})$.   

Initially the q-digest starts out empty with no buckets.  When
processing the next value $v_{ij}$ (the $j$-th value in stream $i$) in
the stream, we first check the q-digest to see if the leaf bucket
$[v_{ij},v_{ij}]$ exists.  If it exists, then we insert the stream id
$i$ into the CM sketch for that bucket and increment the counter for
that bucket by 1.  If that bucket does not exist, then we create a new
$[v_{ij},v_{ij}]$ bucket and insert the pair into that bucket.  After
adding the pair, we carry out a \compress operation on the q-digest.
As more and more data is inserted into the \varb structure, the
buckets are automatically merged and reorganized by the \compress
operation.  This adaptive bucket structure is the reason for the name
\varb.

Given this data structure, let us now see how to approximate the median 
value of a stream $i$.  First we estimate $n_i$ by taking the union of all
the CM sketches in each bucket and then querying it for the value
of $n_i$.  We then do a post-order traversal of the q-digest tree and 
merge the CM sketches of all the buckets visited.  As we merge the CM
sketches, we check the value of the count for stream $i$.  Suppose 
when we merge bucket $b$ to the unified CM sketch, the number of values in 
the unified CM sketch exceeds $n_i/2$.  Then we report the right edge of 
the bucket $b$ as our estimate for the median.   

This estimate has three sources of error. The first error is in the
estimate of $n_i$ itself, which is at most $\eps n$ using the CM sketch
error bound. The second source of error is from estimating the value
of $n_i/2$ while taking union of multiple CM sketches and this error
is again $\eps n$.  The third error is from the error in count of $i$
in bucket $b$.  This error arises from the fact that any value which
is counted in bucket $b$ can be counted on its ancestors as well,
because the bucket $b$ overlaps with all its ancestors.  Since there
are at most $\log U$ ancestors and every ancestor has count at most
$n\rho/\log U$ (eqn.~\ref{eqn:error}), the
total error is $n\rho$ from this source.  Therefore the total error is
$2\eps n + \rho n$.  By rescaling 
$\eps$ by a factor of two, and setting $\rho = \eps$, we arrive at 
the following theorem.

\begin{theorem}
The {\varb} is a data structure of size
$O(\eps^{-2} \log U \log \delta^{-1} )$ that,
with probability at least $1-\delta$, can find the top $k$ streams in a set 
of $m$ streams by average, median, or any quantile value, with (additive) 
rank approximation error $\eps n$, where $n$ is the total size of all the
streams in the braid. 
\end{theorem}

Similarly, we can find top $k$ streams using other measures such as 95th 
percentile, average, etc. 

\input{expts.tex}

\section{Conclusion}
\label{sec:conc}

We investigated the problem of tracking outlier streams in a large set
(braid) of streams in the one-pass streaming model of computation,
using a variety of natural measures such as average, median, or
quantiles.  These problems are motivated by monitoring of performance
in large, shared systems.  We show that beyond the simplest of the
measures (max or min), these problems immediately become provably hard
and require space linear in the braid size to even approximate. It
seems surprising that the problem remains hard even for such minor
extensions of the max as the ``second maximum'' or the spread ($\max
-\min$), or that even highly structured streams with the round robin
order remain inapproximable.  We also propose two heuristics,
{\expb} and {\varb}, analyzed their performance guarantees and
evaluated their empirical performance.

There are several directions for future work. For instance,
we observed that the different Count-Min sketches are used quite unevenly. 
Some sketches are populated to the point of saturation, making their error 
estimates quite bad while others are hardly used.  This suggest that one 
could improve the performance of our data structures by an adaptive 
allocation of memory to the different sketches so that heavily trafficked 
sketches receive more memory than others.
%


\bibliographystyle{plain} \bibliography{biblio}

\end{document}

%% file: expts.tex
\section{Experimental Results}
\label{sec:exp}

In this section, we discuss our empirical results. We implemented both of 
our schemes \expb and \varb and evaluated them on a variety of datasets. In all of our experiments we found that 
\varb consistently outperformed \expb and the main advantage of \expb 
is its somewhat smaller memory usage. (However, the memory advantage is at 
most a factor of 2.) Therefore, we report all performance numbers 
for the \varb only and later show one experiment comparing the relative
performance of the two schemes.

We focused on three most common statistical 
measures of streams: the median, the 95th percentile, and the average value. 
Our goal in these experiments was to evaluate the effectiveness of these 
schemes in extracting the top $k$ streams using these measures. We used 
the following performance metrics for this evaluation.

\begin{enumerate}

\item \textbf{Precision and Recall}: 
	Precision is the most basic measure for quantifying the
	success rate of an information retrieval system.
	If $\mathcal{S}(k)$ is the true set of top $k$ streams under 
	a weight function and $\mathcal{S}^\prime(k)$ is the set of 
	streams returned by our algorithms, then the precision at $k$ $P(k)$ 
	of our scheme is defined as

  \begin{equation}
    P(k) \:=\: \frac{|\mathcal{S}(k) \cap \mathcal{S}^\prime(k)|}{|\mathcal{S}(k)|}
	= \frac{|\mathcal{S}(k) \cap \mathcal{S}^\prime(k)|}{k}.
  \end{equation}

	Thus, precision provides the relative measure of how many of
	the top $k$ are found by our scheme. The precision values always lie 
	between $0$ and $1$, and closer the precision to $1.0$, the better the algorithm.
        We note that in this particular case, precision is the same as recall, defined as 

  \begin{equation}
    R(k) \:=\: \frac{|\mathcal{S}(k) \cap \mathcal{S}^\prime(k)|}{|\mathcal{S}^\prime(k)|}
	= \frac{|\mathcal{S}(k) \cap \mathcal{S}^\prime(k)|}{k}.
  \end{equation}
  
  since $|\mathcal{S}(k)|=|\mathcal{S}^\prime(k)|$. 

\item \textbf{Distortion}: 
	The precision is a good measure of the fraction of top $k$ streams 
	found by our algorithm, but it fails to capture the ranking of those
	streams. For example, suppose we have two algorithms, and both correctly
	return the top 10 streams but one returns them in the correct rank
	order while the other returns them in the reverse order.
	Both algorithms enjoy a precision of 1.0 but clearly the second
	one performs poorly in its ranking of the streams.
	Our \emph{distortion} measure is meant to capture this ranking quality.
	Suppose for a stream $S_i$ the true rank is $r(S_i)$, while our
	heuristic ranks it as $r' (S_i)$. Then, we define the (rank)
	distortion for stream $S_i$ to be

	\[ d_i \:=\:  \left\{  \begin{array}{l}
				r(S_i)/r'(S_i), \quad \mbox{if}~ r(S_i) \geq r'(S_i) \\
				r'(S_i)/r(S_i), \quad \mbox{otherwise} .
				\end{array} \right.
	\]

	The overall distortion is taken to be the average distortion for the $k$ 
	streams identified by our scheme as top $k$.
	The ideal distortion is 1, while the worst distortion can be
	$\Theta (m/k)$, where $m$ is the size of the braid: this happens when the 
	algorithm ranks the bottom $k$ streams as top $k$, for $k \ll m$.
	Thus, smaller the distortion, the better the algorithm.

\item \textbf{Value Error}: 
	Both precision and distortion are purely rank-based measures,
	and ignore the actual values of the weight function  $\lambda (S)$.
	In the cases when data is clustered, many streams can have roughly
	the same $\lambda$ value, yet be far apart in their absolute ranks.
	Since in many monitoring application, we care about streams
	with large weights, a user may be perfectly satisfied with any stream
	whose weight is close to the weights of true top $k$ streams.
	With this motivation, we define a value-based error metric, as follows.
  	Suppose the true and approximate streams at rank $k$ are $S_k$ and 
	$S_k^\prime$.  Then the \emph{relative value error} $e(k)$ is defined 
	as
  \begin{equation}
    	e(k) = \left| \frac{\lambda(S_k)-\lambda(S_k^\prime)}{\lambda(S_k)}\right|
\label{eqn:value-err}
  \end{equation}

	The average value error $e$ for the top $k$ is then defined as the
	average of $e(k)$ over all $k$ streams.

\item \textbf{Memory Consumption}: The space bounds that our theorems
  	give are unduly pessimistic.  Therefore we also empirically
	evaluated the memory usage of our scheme.
\end{enumerate}

We generate several synthetic data sets using natural distributions
to evaluate the performance and quality of our algorithms.
In all cases, we use 1000 streams, with about 5000 items each,
for the total size of all the streams 5M.
In all cases, the values within each stream are distributed using a
Normal distribution with variance $U/20$. The mean values for each distribution are
picked by an inter-stream distribution, for which we try 3 different
distributions: uniform, outlier, and normal.

\begin{itemize}
\item In the \textbf{Uniform distribution}, we pick values uniformly at random
  from the range $U = [1, 2^{16}]$, and each such value acts as the
  mean $\mu_i$ for stream $S_i$. 
  
\item In the \textbf{Outlier distributions}, we choose 900 of the streams with values
  in the range $[0, 0.6U]$ and the remaining 100 streams in the range
  $[aU, (a+0.2)U]$, with $a<1$, for different values of $a$. 

\item In the \textbf{Normal distribution}, the values are chosen from a normal
  distribution with mean $2^{15}$, and standard deviation $2^{14}$.
\end{itemize}

\subsection{Precision}

\begin{figure}
\subfigure[Precision for $\lambda$=average]{
  \includegraphics[width=0.45\textwidth]{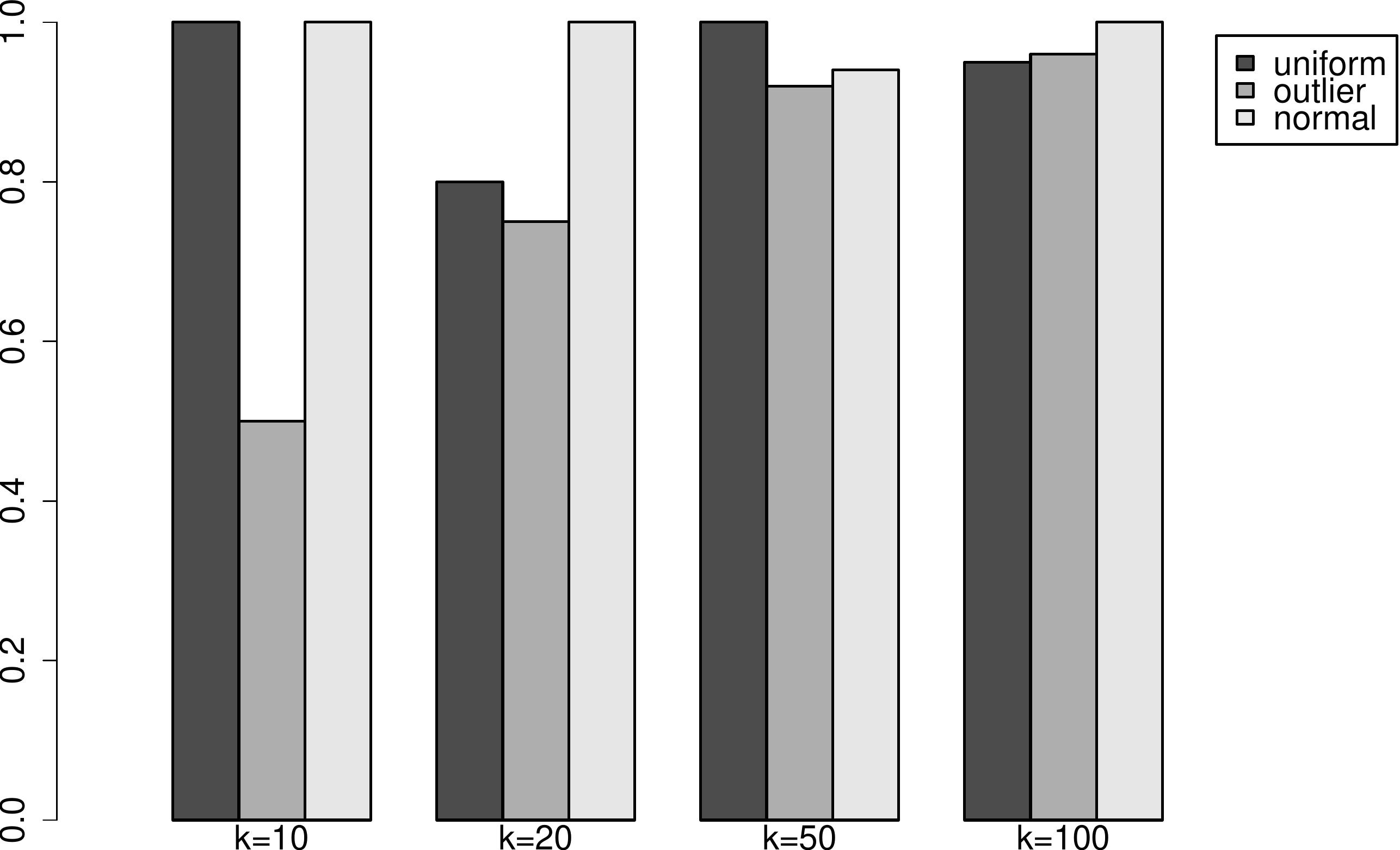}

\bigskip
}

\subfigure[Precision for $\lambda$=median]{
  \includegraphics[width=0.45\textwidth]{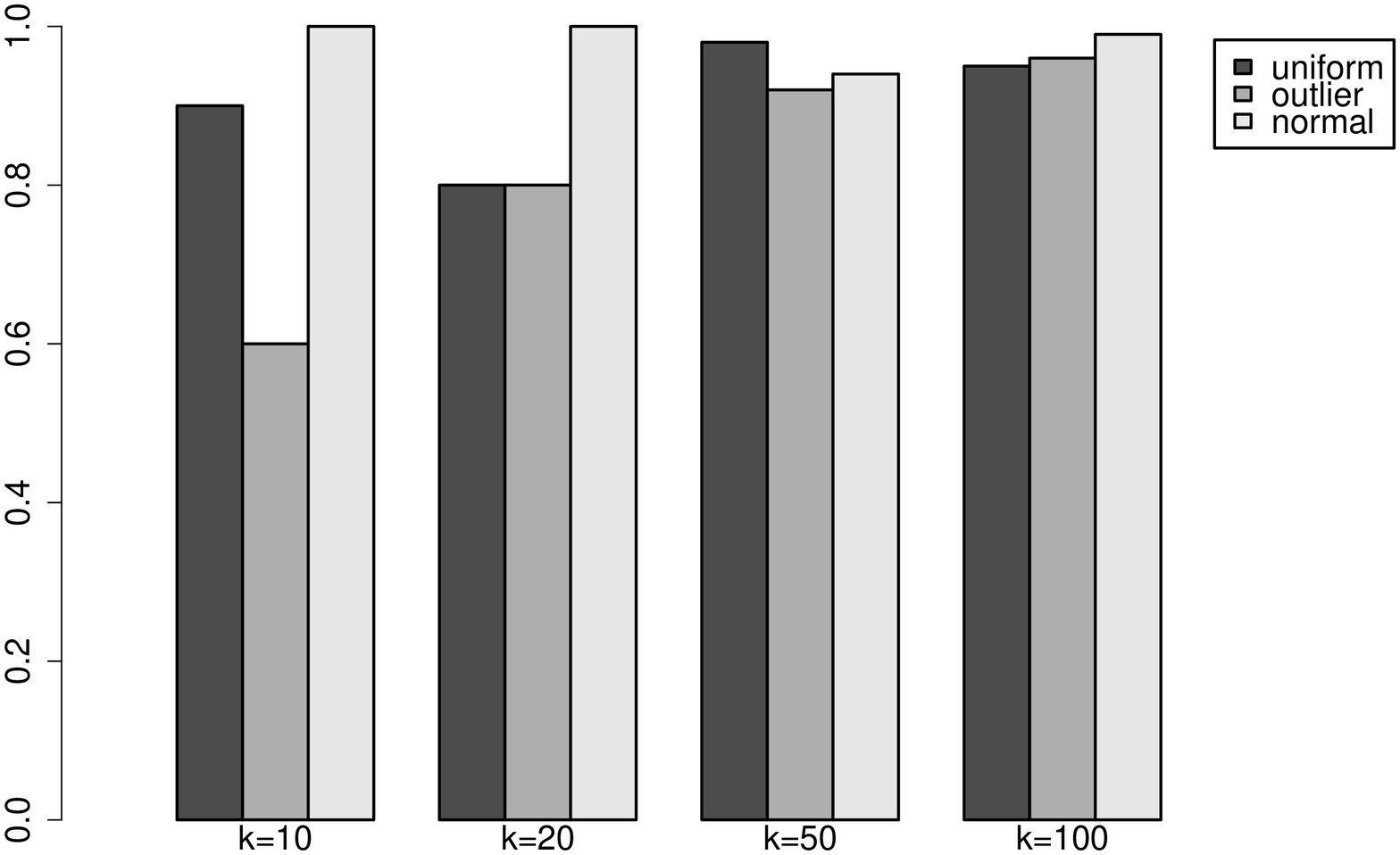}

\bigskip
}

\subfigure[Precision for $\lambda$=95th percentile]{
  \includegraphics[width=0.45\textwidth]{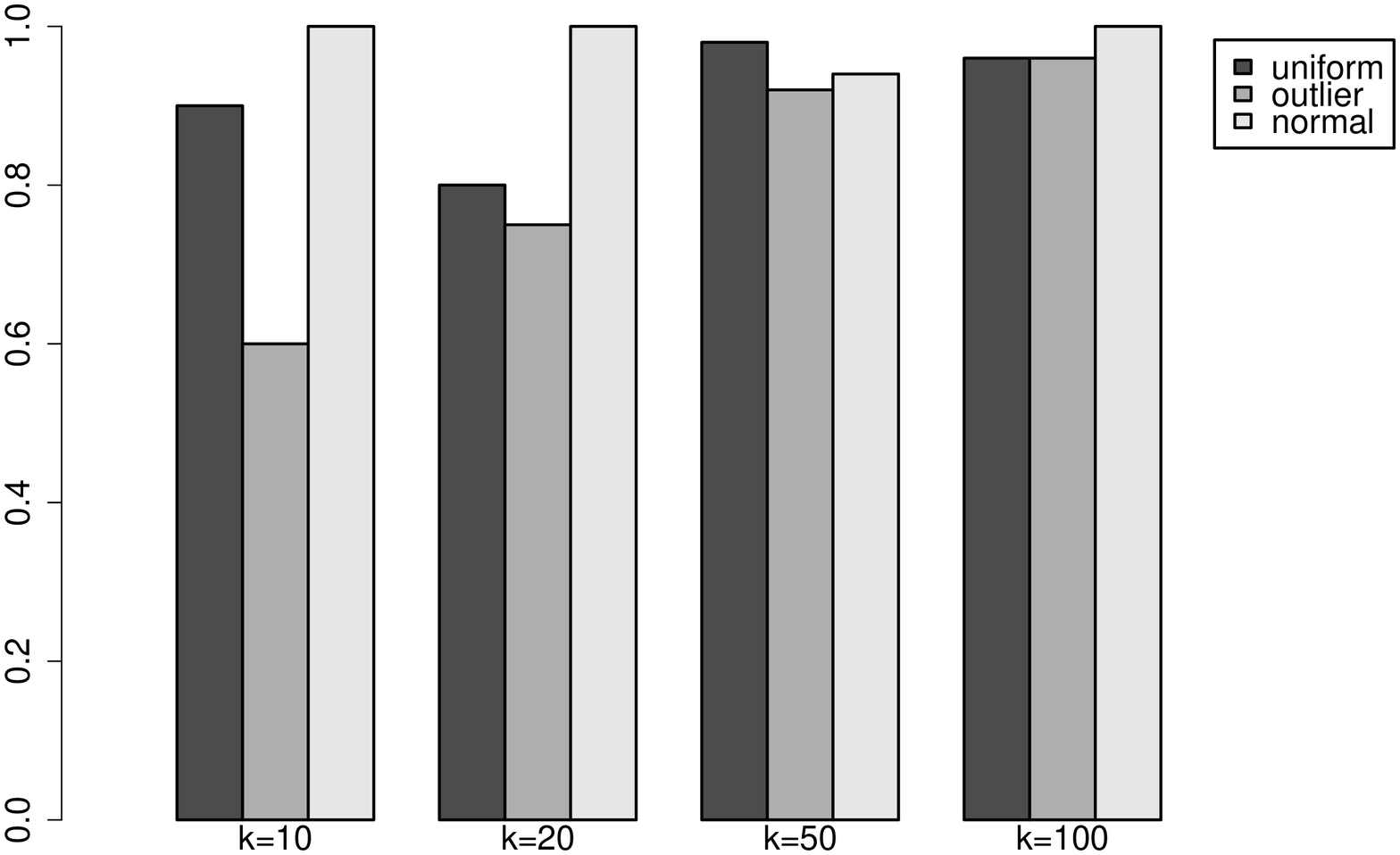}

\bigskip
}
  \caption{The precision quality of \varb as a function of $k$, for
	the three choices of $\lambda$: average, median, and 95th percentile.}
\label{fig:precision-vs-k-varb}    
\end{figure}

Our first experiment evaluates the precision quality: how many of 
the true top $k$ streams are correctly identified by our algorithms. 
Figure~\ref{fig:precision-vs-k-varb} shows the results of this experiment,
where for each data set (of 1000 streams), we asked for the top $k$,
for $k=10, 20, 50, 100$. In Figure~\ref{fig:precision-vs-k-varb} the outlier distribution has parameter $a=0.8$.
We evaluated the precision for each of the three choices of $\lambda$:
average, median, and 95th percentile.
As the figure shows, the precision quality of \varb begins to approach
100\% for $k \geq 50$.
The pattern is similar for average, median, or the 95th percentile.
The precision achieved by \varb on the outlier distribution degrades as the parameter $a$ decreases.
as it can be seen in Figure~\ref{fig:prec_varb_outlier}. This behavior
is easily explained by the fact that the parameter $\alpha$ sets the separation
between outlier and non-outlier streams. The smaller $\alpha$ is, the
fuzzier becomes the separation, therefore \varb success rate in
identifying outliers decreases.

\begin{figure}
  \includegraphics[width=0.45\textwidth]{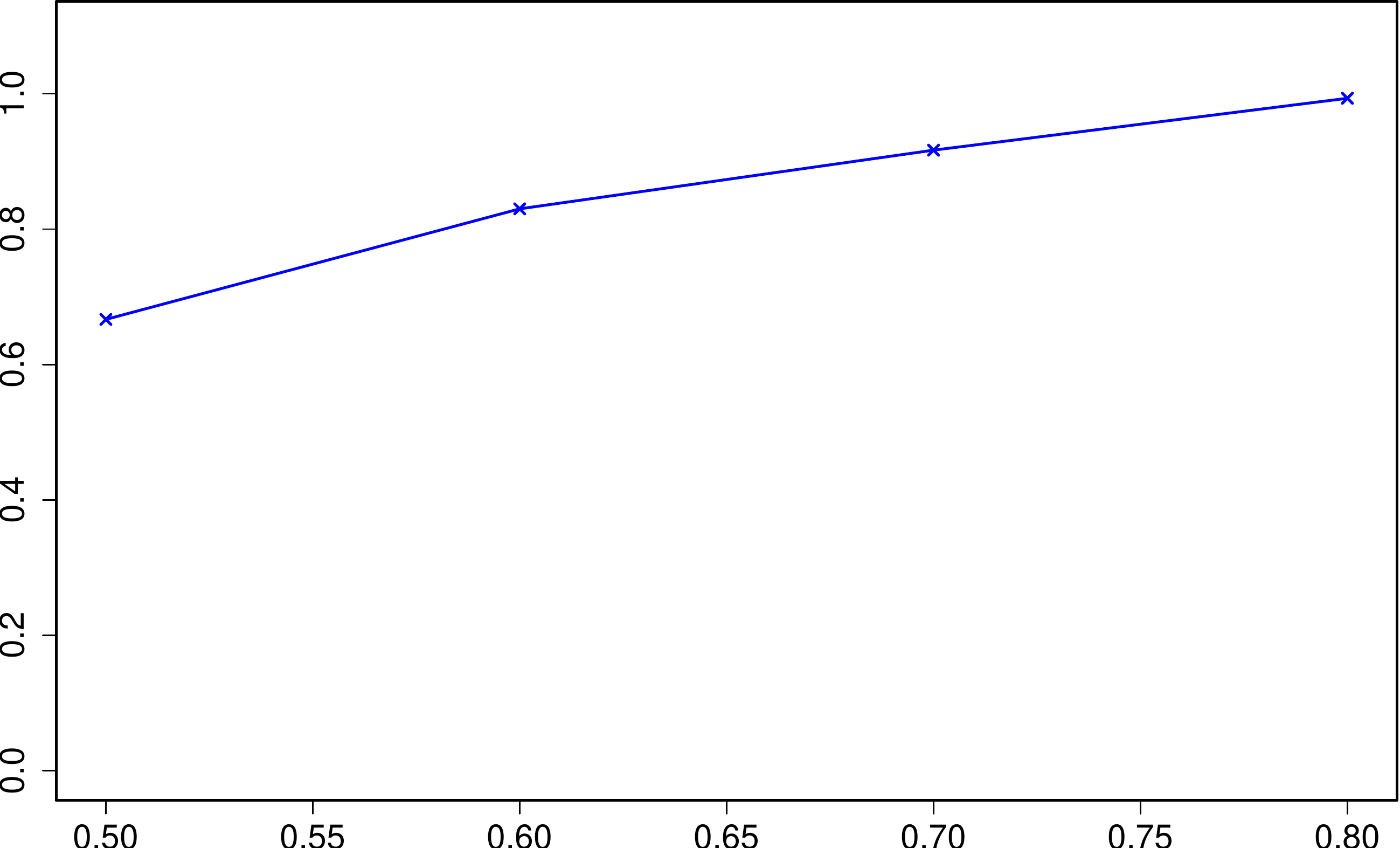}
  \caption{Precision achieved by \varb for different $a$ in the outlier distribution, top-100, $\lambda$=median}
\label{fig:prec_varb_outlier}
\end{figure}

\subsection{Distortion Performance}

\begin{figure}
\subfigure[Distortion for $\lambda$=average]{
  \includegraphics[width=0.45\textwidth]{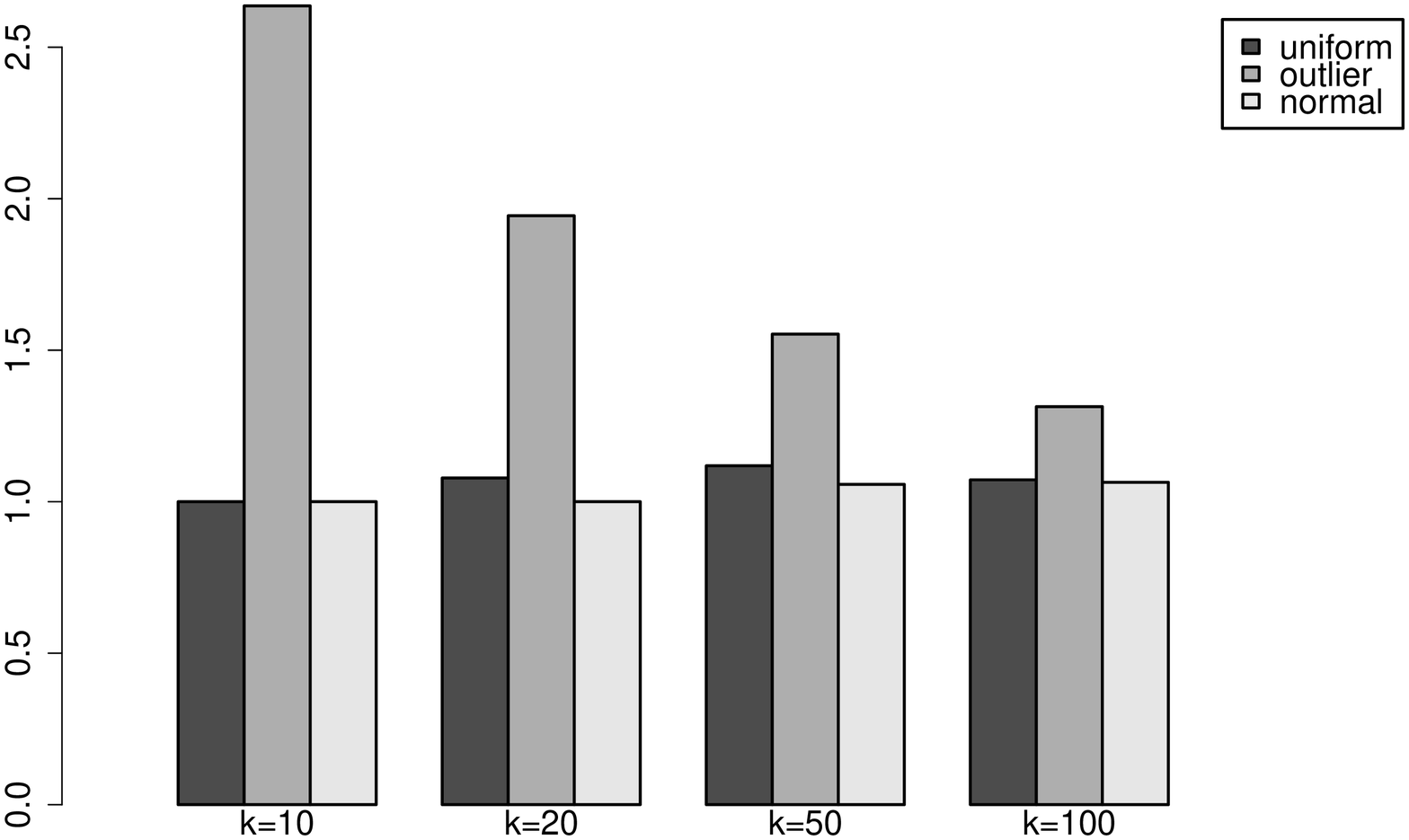}

\bigskip
}
\subfigure[Distortion for $\lambda$=median]{
  \includegraphics[width=0.45\textwidth]{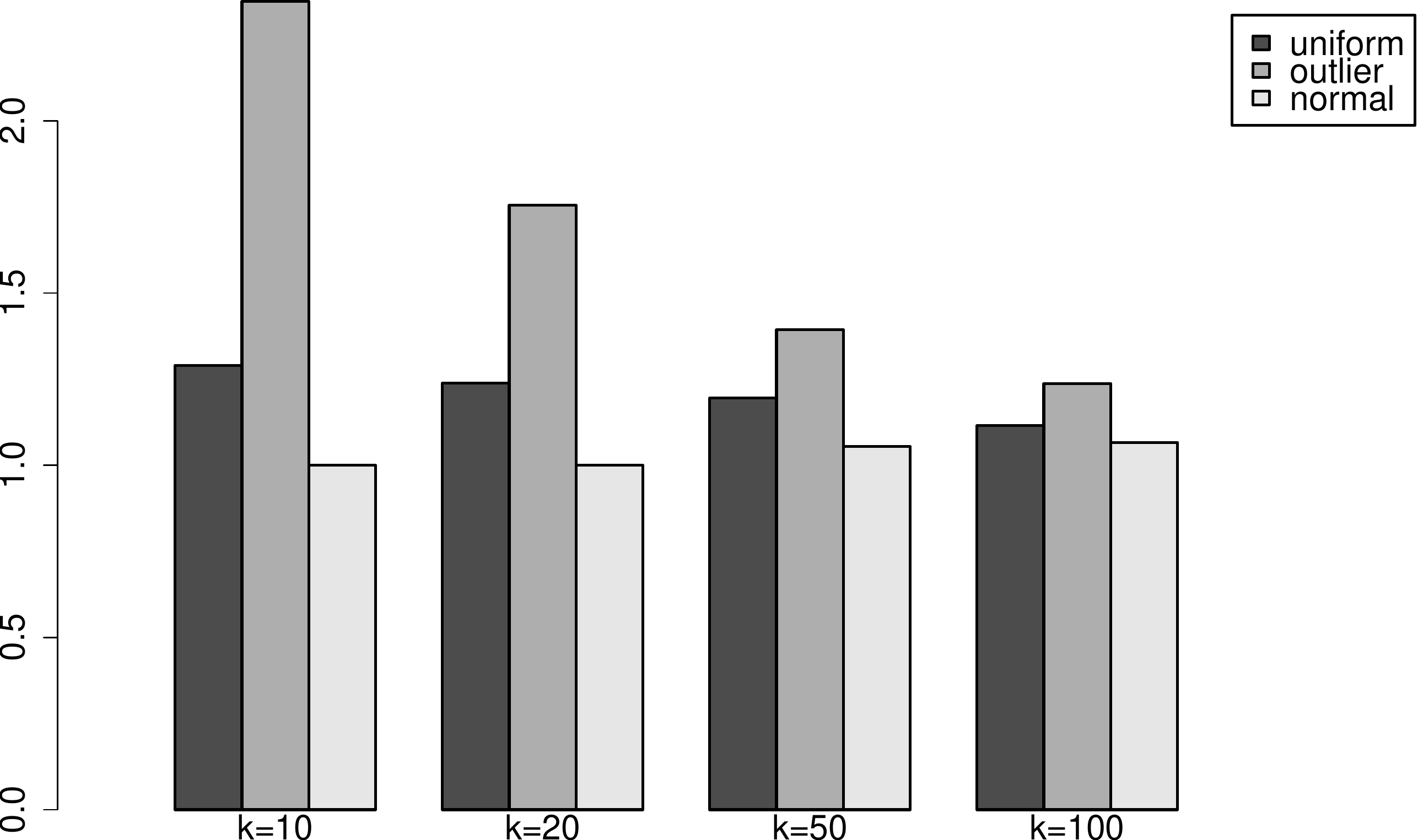}

\bigskip
}
\subfigure[Distortion for $\lambda$=95th percentile]{
  \includegraphics[width=0.45\textwidth]{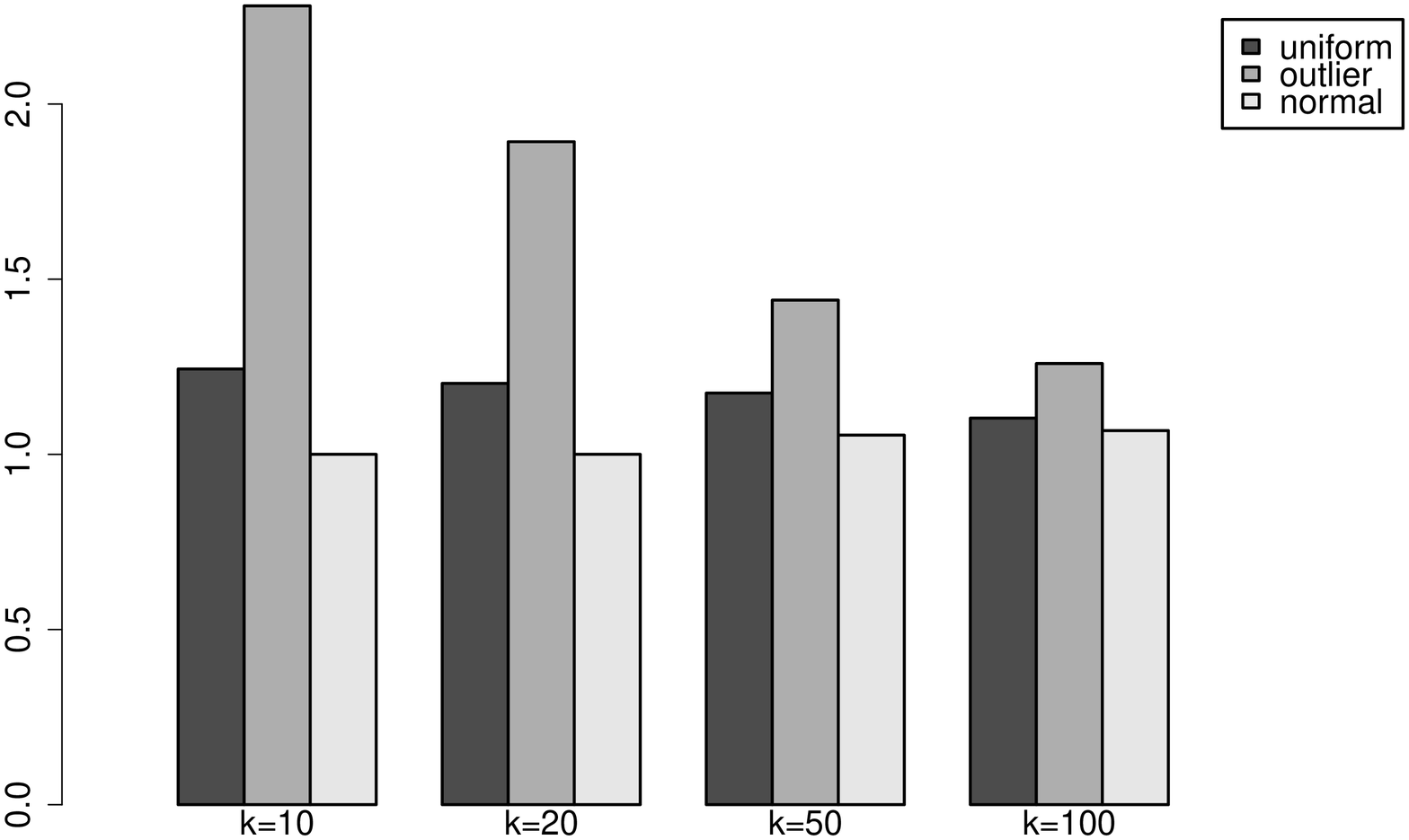}

\bigskip
}
  \caption{The distortion performance of \varb as a function of $k$, 
	for the three choices of $\lambda$: average, median, and 95th percentile.}
\label{fig:dist-vs-k-varb}    
\end{figure}

Our second experiment measures distortion in ranking the top $k$ streams, 
under the three weight functions average, median, and 95th percentile.
The results are shown in Figure~\ref{fig:dist-vs-k-varb}.
For all three data sets, distortion is uniformly small (between 1 and 4), 
even for $k$ as small as 10, and it actually drops to the range 1--2 for $k \geq 50$.

\subsection{Average Value Error}

\begin{figure}
\subfigure[Average value error for $\lambda$=average]{
  \includegraphics[width=0.45\textwidth]{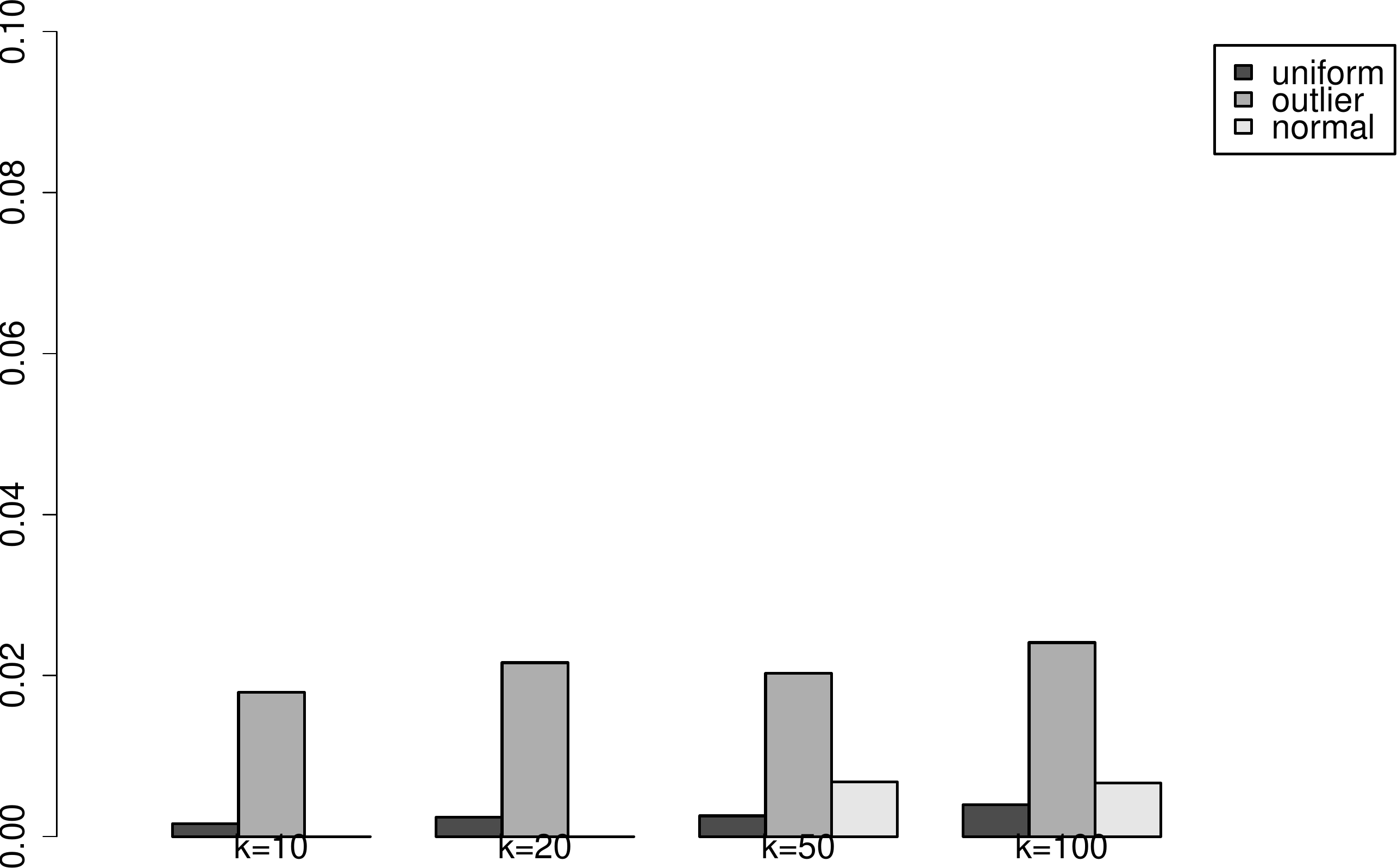} \\[10ex]

\bigskip
}
\subfigure[Average value error for $\lambda$=median]{
  \includegraphics[width=0.45\textwidth]{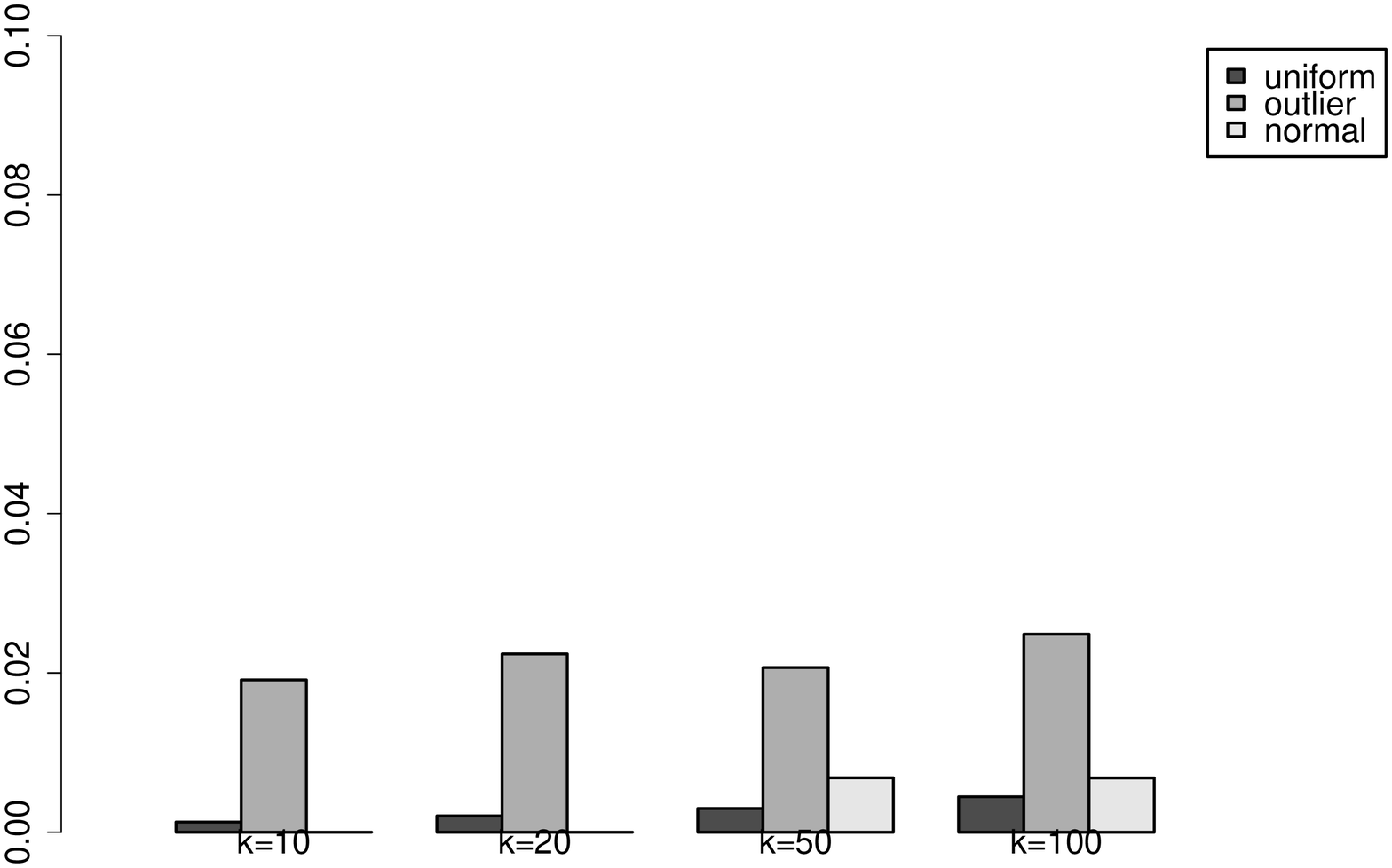} \\[10ex]

\bigskip
}
\subfigure[Precision for median $\lambda$=95th percentile]{
  \includegraphics[width=0.45\textwidth]{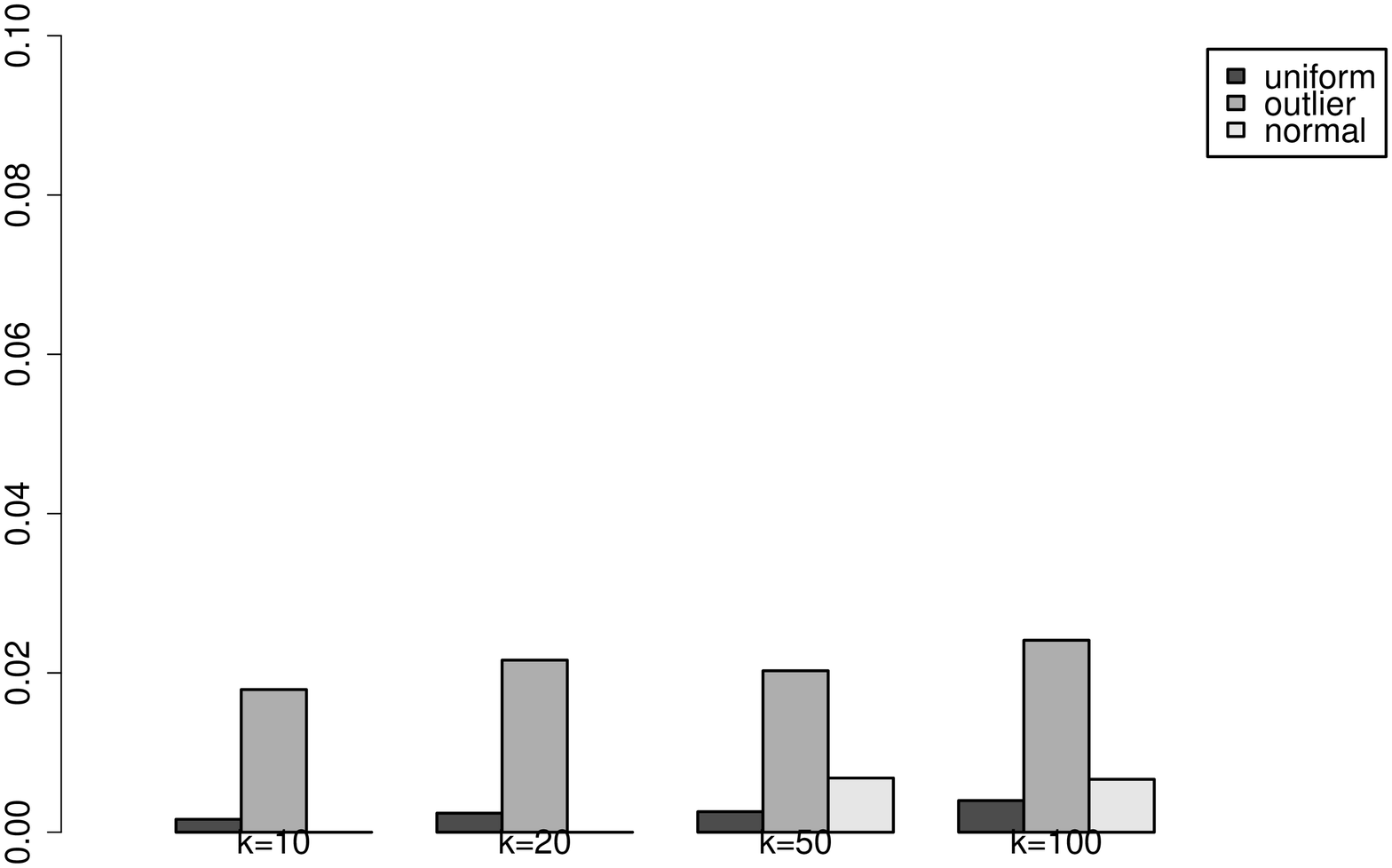}

\bigskip
}
  \caption{The average value error of \varb as a function of $k$,
	for the three choices of $\lambda$: average, median, and 95th percentile.}
\label{fig:aver-vs-k-varb}    
\end{figure}


The previous two experiments have attempted to measure the quality
of our scheme using a rank-based metric.  In this section, we consider
the performance using the value error, as defined earlier. 
The results are shown in Figure~\ref{fig:aver-vs-k-varb}.
For all distributions, the relative error in the value of the top $k$ streams is quite small: of the order
of 1--2\%. Thus, even when the algorithm finds streams outside the true
top $k$, it is identifying streams that are close in value to the true
top $k$. This is especially encouraging because in data without
clear outliers, the meaning of top $k$ is always a bit fuzzy.

\subsection{ExponentialBucket vs. VariableBucket}

In our experiments, we tried both our schemes, \expb and \varb , on all
the data sets, but due to the space limitation, we reported all the
results using \varb only. In this section, we show one comparison
of the two schemes to highlight their relative performance.
Figure~\ref{fig:aver-vs-k-avg-expb} shows the results for the
precision using the median weight, for all three data sets.
The bottom figure is the same one as in Figure~\ref{fig:precision-vs-k-varb}
(middle), while the top one shows the performance of \expb for this experiment.
One can see that in general \varb delivers better precision than \expb.
This was our observation in nearly all the experiments, leading us to
conclude that \varb has better precision and error guarantees than
\expb. This is also consistent with out theory, where we found that
\varb can be shown to have bounded rank error guarantee while
\expb could not.
On the other hand, \expb does have a memory advantage: its data structure
consistently was more space-efficient that that of \varb, so when space
is a major constraint, \expb may be preferable. However, the space usage of
\varb itself is not prohibitive, as we show in the following experiment.

\begin{figure}
\centering
\subfigure[Precision for median, \expb]{
\includegraphics[width=0.45\textwidth]{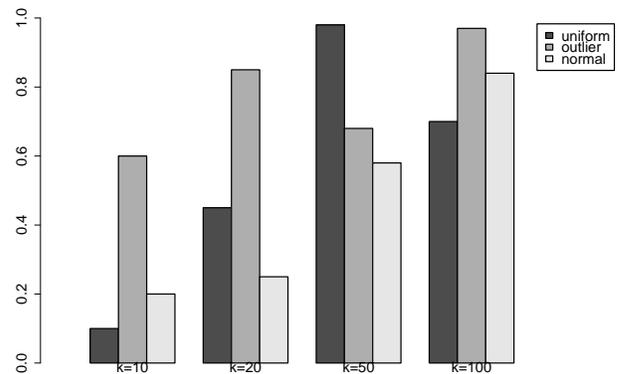}

\bigskip
}
\subfigure[Precision for median, \varb]{
\includegraphics[width=0.45\textwidth]{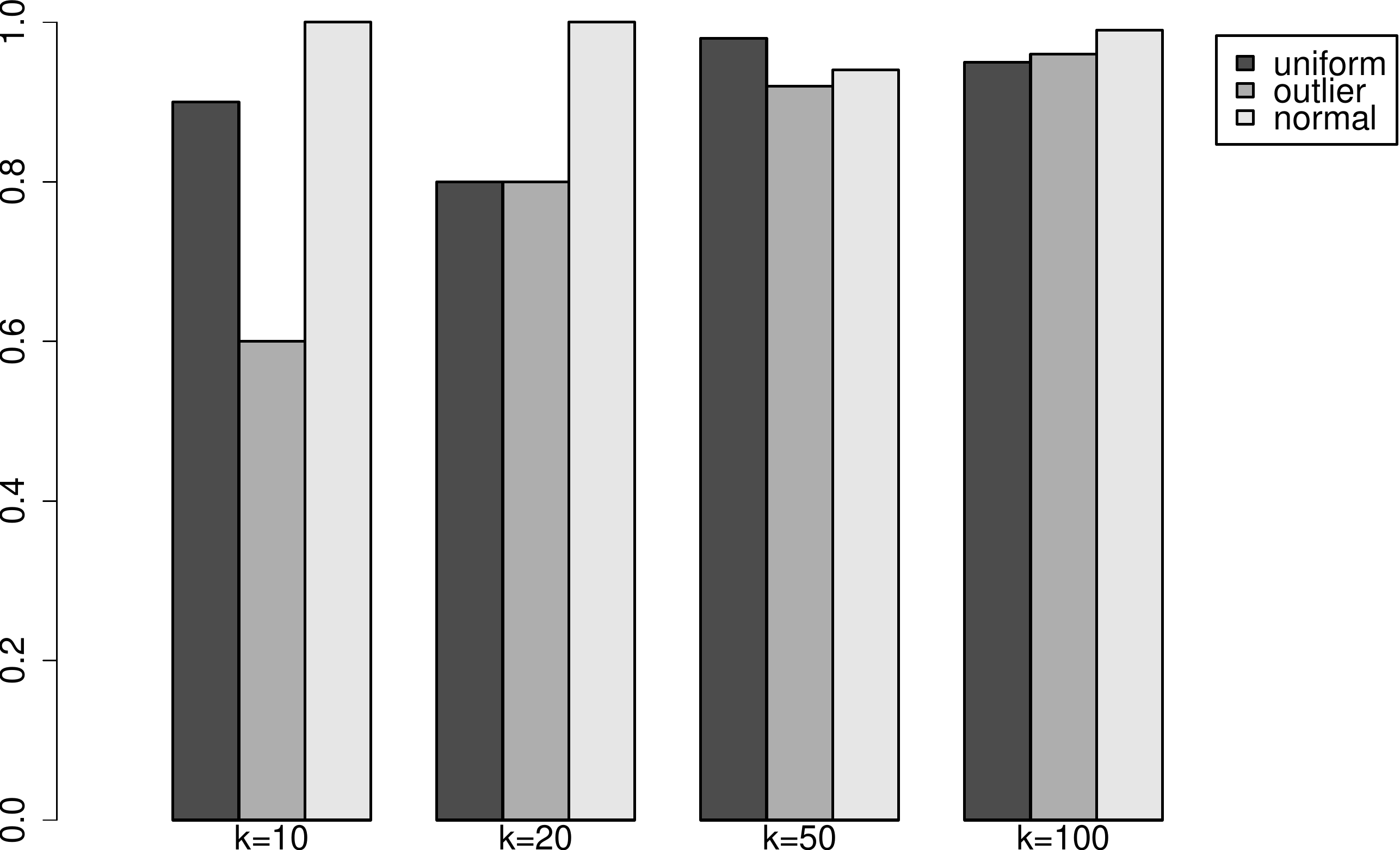}

\bigskip
}
\caption{A comparison of \expb with \varb. The \expb is more space-efficient
  but consistently does worse than \varb. This experiment shows,
  side-by-side, the results of the precision quality experiment
  for the median weight, using the two schemes.}
\label{fig:aver-vs-k-avg-expb}
\end{figure}



\subsection{Memory Usage}

In this experiment, we evaluated how the memory usage of \varb 
scales with the size of the braid. In theory, the size of \varb does not
grow with $m$, the number of streams, or the size of individual streams.
However, theoretical bounds on the space size are highly pessimistic, so
used this experiment to evaluate the space usage in practice.
In our implementation of \varb we used a Count-Min sketch with
depth 64 and width 64. We then built \varb for number of streams
varying from $m=1000$ to $m=10,000$, and Figure~\ref{fig:size}
plots the memory usage vs. the number of streams. As predicted, the
data structure size remains virtually constant, and is about 2 MB.

\begin{figure}
  \includegraphics[width=0.4\textwidth]{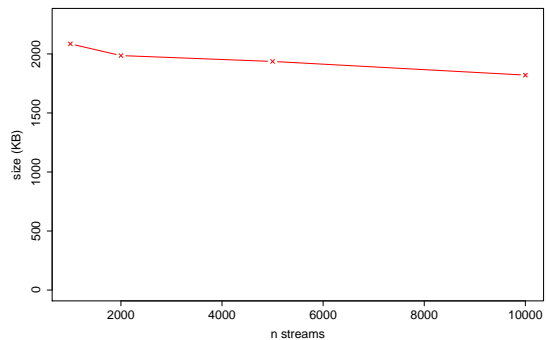}
  \caption{Data structure size as a function of the braid size.}
\label{fig:size}
\end{figure}